\documentclass[11pt]{article}%
\usepackage{hyperref}
\usepackage{amsmath}
\usepackage{graphics}
\usepackage{color}
\usepackage{epsfig}
\usepackage{graphicx}%
\usepackage{amsfonts}%
\usepackage{amssymb}
\usepackage{setspace}
\usepackage[margin=1in]{geometry}
\usepackage{comment}
\usepackage{listings}
\usepackage{color}
\newtheorem{theorem}{Theorem}[section]

\newtheorem{definition}{Definition}[section]

\newtheorem{lemma}{Lemma}[section]

\newcommand{\qed}{\hfill \mbox{\raggedright \rule{2mm}{3mm}}}
\newenvironment{proof}{\noindent{\bf Proof.}}{\qed}

\newcommand{\cfl}{{\sc Cfl}}

\date{}

\allowdisplaybreaks

\begin{document}

\title{
Exponential lower bounds on the size of approximate formulations in the natural encoding  for Capacitated Facility Location\thanks{
This research has been co-financed by the European Union (European
Social Fund -- ESF) and Greek national funds through the Operational
Program ``Education and Lifelong Learning'' of the National Strategic
Reference Framework (NSRF) - Research Funding Program:
``Thalis. Investing in knowledge society through the European Social Fund''.}
}

\author{Stavros G. Kolliopoulos\thanks{Department of Informatics and
Telecommunications, National and Kapodistrian 
University of Athens, Panepistimiopolis Ilissia, Athens
157 84, Greece; (\texttt{sgk@di.uoa.gr}).}   
\and Yannis Moysoglou\thanks{ 
Department of Informatics and
Telecommunications, National and Kapodistrian 
University of Athens, Panepistimiopolis Ilissia, Athens
157 84, Greece; (\texttt{gmoys@di.uoa.gr}). } }
\maketitle

\thispagestyle{empty}



\maketitle

\begin{abstract}
The metric capacitated facility location is a well-studied problem for
which, while  constant factor  approximations are known,  no efficient
relaxation  with  constant  integrality  gap is  known.  The  question
whether there is such a relaxation is among the most important open
problems of approximation algorithms \cite{ShmoysWbook}.

In this  paper we show that,  if one is restricted  to linear programs
that use the natural 
encoding for facility location, at least an exponential number of constraints
is needed  to achieve a  constant gap. Our  proof does not  assume any
special property of the relaxation such as locality or symmetry.
\end{abstract}

\section{Introduction}

In recent years there has been an increasing interest in 
characterizing the strength of linear programming relaxations
for approximating combinatorial optimization problems. In the seminal
paper of Arora et al.  \cite{AroraBLT06} the integrality gap of general
families   of   relaxations  for   the   Vertex   Cover  problem   was
studied.  These families include  relaxations with  local constraints,
relaxations with low-defect inequalities and those relaxations
obtained  after   $O(\log  n)$  rounds   of  the  Lov\'{a}sz-Schrijver
hierarchy (\cite{LovaszS91}).

Subsequently, the  idea of fooling  local constraints was  extended to
deriving lower  bounds on  the number of  levels of  the Sherali-Adams
hierarchy (\cite{SheraliA94}) starting with \cite{FernandezdlVKM07}.  This
framework has fueled several hierarchy-based gap bounds
of the past years; it was given the name ``from-local-to-global'' in
\cite{GeorgiouMagen}.

Recently, limitations on the approximation strength of extended
formulations were shown for the maximum clique problem
\cite{BraunFPS12,BravermanM13}.  In \cite{ChanLRS13} it is
proved that in terms of approximation, LPs of size $O(n^k)$ are
exactly as powerful as $O(k)$-level  relaxations of Sherali-Adams
hierarchy for maximum constraint satisfaction problems.

The  {\em metric capacitated  facility location}  problem (\cfl)  is a
well-studied problem for which, while constant factor approximations 
are known \cite{BansalGG12,AggarwalLBGGJ12}, 
no efficient LP relaxation with constant integrality gap is
known. The question whether  such a relaxation exists is among the most
important     open     problems     in    approximation     algorithms
\cite{ShmoysWbook}. An instance $I$ of \cfl\ is defined as follows.  We are given a set
$F$ of facilities and set $C$  of clients. We may open facility $i$ by
paying its opening cost $f_i$ and we may assign client $j$ to facility
$i$ by  paying the connection  cost $c_{ij}$. The 
latter costs satisfy the following variant of the triangle inequality:
$c_{ij} \leq c_{ij'} + c_{i'j'} + c_{i'j}$ for any $i, i'\in F$ and
$j, j' \in C.$
We  are asked to  open a
subset $F' \subseteq F$ of the facilities and assign each client to an
open  facility.  The  goal  is  to  minimize  the  total  opening  and
connection cost.

Apart from some previous work of the authors
\cite{KolliopoulosM13,KolliopoulosM14}, no other progress
towards the resolution of the question about the linear programming
approximability of \cfl\ has been made. In this paper we give further
negative evidence for this  notorious open problem by ruling out all
polynomially-sized relaxations that use the natural encoding with  
facility opening and client assignment variables.  This is  a quite
general family of relaxations that in the case of \cfl, or network
design problems in general, has been the focus of attention over the 
years.

\subsection{Our results}

In this paper we show unconditionally that at least an exponential number of
constraints is necessary 
for a relaxation of \cfl, that uses the natural encoding and has 
constant integrality gap. We do not make any assumptions on
the structure of the constraints such as locality or symmetry.

Our proof,  described at  a high level,  uses a simple  yet insightful
counting argument. We identify a large number of (fractional) vectors
such that for
each one of them, call it $s,$ there is an admissible cost vector for which 
$s$ induces a cost that is $o(\mbox{OPT}),$ where OPT is the cost of
the optimal integer solution with respect to the same cost vector.  Then we
show that  an arbitrary  valid inequality cannot  be violated  by more
than a small number of such fractional vectors. Thus by using the union bound we get
that  a large  number  of  inequalities is  needed  to separate  those
problematic  points from  the  feasible region.  A  similar idea  was
independently used by Kaibel  and Weltge in \cite{KaibelW13} to derive
lower bounds on the number of  facets of a polyhedron which contains a
given set  $X$ of integer  points and whose  set of integer  points is
$\operatorname{conv}(X)\cap  \mathbb{Z}^d$. We  note however  that for
problems  such  as  facility  location, the  known  polynomially-sized
relaxations   already   have    the   aforementioned   property.   Our
implementation  of the counting  argument is  more general  and allows
proofs for bounds on approximate relaxations that achieve a given gap quality
$g$.

A challenge in our proof is  showing  membership of a vector in the
convex hull of integer solutions  for \cfl\/. We overcome this problem
by building on the  elegant probabilistic framework that we introduced
in  \cite{KolliopoulosM14}. We  believe that  our techniques  apply to
other problems as well.

\section{The Method}

Here we present in detail our methodology which we will subsequently use
to derive results for metric capacitated facility location.

Let $Q=\{ x \in[0,1]^n \mid Ax\leq b\}$ be a linear relaxation and let $P=Q\cap
\{0,1\}^n$   be   the   set   of   integer  solutions  to   $Q$   and
$\operatorname{conv}(P)$ be the  convex hull of $P$ (we  will also say
that    $\operatorname{conv}(P)$   is   the    corresponding   integer
polytope). Our method consists of the following steps.

We design a family $\mathcal{I}$ of instances parameterized by the
dimension. For  each instance $I\in  \mathcal{I}$ of dimension  $n$ we
define a set of (exponentially many) points in $[0,1]^n$ which we call
the   \emph{core}   of  $I$.   We   denote   the   core  of   $I$   by
$\mathcal{C}_I$. We show that for  each $s \in \mathcal{C}_I$ there is
an admissible cost function $w_s$ such that $w_s^T s=o(Opt_{I,w_s})$
where $Opt_{I,w_s}$ is  the cost of the optimal  integer solution with
respect to $w_s$.

Then we prove that each   inequality $\pi$ of $Q$ can separate at
most      $\lambda$      members      of     $\mathcal{C_I}$      from
$\operatorname{conv}(P)$. We do so by the following argument: let $s_1
\in   \mathcal{C_I}$   be  a   vector   that   $\pi$  separates   from
$\operatorname{conv}(P).$ 
Then we
identify a set $\mathcal{U}\subseteq  \mathcal{C_I}$ such that for any
$s_2\in \mathcal{U}$ a convex combination  $s'$ of $s_1,s_2$ is also a
convex  combination of  integer  solutions. Note  here that  $s_1,s_2$
themselves are  not in $\operatorname{conv}(P)$. Thus  by validity and
by selecting the size of $\mathcal{U}$ to be independent from $s_1$ we
have      that      $\pi$      cannot     separate      more      that
$\lambda=|\mathcal{C_I}|-|\mathcal{U}|$ members of $\mathcal{C_I}$.

By     the     union     bound     we    get     that     at     least
$\frac{|\mathcal{C_I}|}{\lambda}$  inequalities are  needed  to separate
all members of  $\mathcal{C_I}$ from $\operatorname{conv}(P)$. Thus at
least that  many inequalities  are needed to  acquire a  relaxation of
constant gap. To derive exponential  bounds on the size of approximate
relaxations, $\lambda$ is needed to be
$\frac{|\mathcal{C_I}|}{2^{\Omega(n)}}$.

\section{Preliminaries}

In  what  follows,  we  use  the definition  of  a  $\rho$-approximate
relaxation   as    given   by   \cite{BraunFPS12}.    We   note   that
\cite{BraunFPS12} is concerned with extended approximate relaxations
that  use the  {natural encoding},  a  direction that  we do  not
pursue in this work.

Given a combinatorial optimization problem $T$, a {\em linear encoding} of
$T$ is  a pair $(L,O)$  where $L \subseteq  \{0, 1\}^*$ is the  set of
{\em feasible solutions} to the problem and $O\subset \mathbb{R}^*$ is the
set of {\em admissible objective functions.}  An instance of the linear
encoding is  a pair $(d,w)$ where  $d$ is a  positive integer defining
the dimension of the instance  and $w\subseteq O\cap \mathbb{R}^d $ is
the  set  of admissible  cost  functions  for  instances of  dimension
$d$. Solving  the instance  $(d,w)$ means finding  $x \in L  \cap \{0,
1\}^d$ such that  $w^{T}x$ is either maximum or  minimum, according to
the type of  problem under consideration.  Let $P=\operatorname{conv}(
\{x \in \{0, 1\}^d \mid x \in L \})$ be the integer polytope of dimension $d$.

Given a linear encoding $(L,O)$ of a maximization problem, and $\rho \geq
1$,  a  $\rho$-\emph{approximate  formulation  that uses  the  natural
  encoding} is a  formulation $Ax\leq b$ with $x\in \mathbb{R}^d$
such that

\begin{align*}
\max \{w^{T}x \mid Ax\leq b \} \geq & \max \{w^{T}x \mid x \in P\} &
  \mbox{ for all }  w \in \mathbb{R}^d \mbox{ and }\\
\max \{w^{T}x \mid Ax\leq b \} \leq  & \rho \max \{w^{T}x \mid x \in
  P\} & \mbox{ for all }  w \in O \cap \mathbb{R}^d.
\end{align*}
 For a minimization problem,  we require 
\begin{align*}
\min \{w^{T}x \mid Ax\leq b \} \leq & \min \{w^{T}x \mid x \in P\}   
& \mbox{ for all }  w \in \mathbb{R}^d \mbox{ and }\\
\min \{w^{T}x \mid Ax\leq b \}  \geq &  \rho^{-1} \min \{w^{T}x \mid x \in
P\} & \mbox{ for all }  w \in O \cap \mathbb{R}^d.
\end{align*}

\section{Bounds for \cfl\ }\label{section:cfl}

In the case  of \cfl, the
linear encoding $(L, O)$ is defined as follows. For a \cfl\ instance, given
the number $n$ of facilities, the number $m$ of clients, the
capacities    $K\in   \mathbb{R}^n_+$    and    the   demands    $D\in
\mathbb{R}^m_+$, 
we use the variables $y_i,$
$i=1,\ldots,n,$  $x_{ij},$  $i=1,\ldots,n,$  $j=1,\ldots,m$  with  the
usual  meaning of facility opening and client assignment respectively. 
The set of feasible solutions $(y,x)$ is defined 
in the obvious manner. 
Thus for dimension $d=n + nm,$ 
$L \cap \{0,1\}^d$ is completely determined by the quadruple
$(n,m,K,D).$ The
set of admissible objective functions $O \cap \mathbb{R}^{n+nm}$ is the set of pairs
$({\bf f}, {\bf c})$ where ${\bf f} \in \mathbb{R}^n_+$ are the facility
opening costs and ${\bf c}=[c_{ij}] \in \mathbb{R}^{nm}_+$ are 
 connection costs  that satisfy  $c_{ij}\leq
c_{i'j}+c_{i'j'}+c_{ij'}$.   

In our proof we will consider feasible sets of the form $(n,m, U{\bf
  1},{\bf 1}),$ 
i.e., with uniform
capacities $U >0,$ and unit demands. Therefore the triple $(n,m,U)$ is
sufficient description. 
Furthermore, it will be convenient to deviate from the convention that the
number  of facilities is  $n$ --  this is  to simplify  the expressions
appearing through the proof. Let the number of facilities be
$n^2,$ the number of clients be $an^4$ for some integer $a\geq 2$
and the capacity $U$ of each facility be $n^3.$ Thus for a given $n,$
the feasible set is uniquely determined by the triple
$(n^2,an^4,n^3).$ To  avoid cumbersome expressions,  we slightly abuse
terminology and refer to such a
triple as 
an {\em instance $I(n^2,an^4,n^3).$} We denote for the instance
in question the set of
facilities by $F$ and the set of clients by $C$.

We first describe the core
$\mathcal{C}_I$ of the instance $I(n^2,an^4,n^3)$.
\begin{definition}
The  core $\mathcal{C}_I$  of  the instance  $I(n^2,an^4,n^3)$ is  the
following set  of $(y,x)$ vectors.  $\forall k,l \subset F$ with
$|k|,|l| =  n$ and  $k\cap l  =\emptyset$ and for  a set  $C_{k,l}$ of
clients with $|C_{k,l}|=Un+1$ we define a vector $s_{k,l}$ such that: 
(1) $y_{i}=1, \forall i\in k$, $y_{i}=\frac{10}{n^2}, \forall i\in l$,
$y_{i}=1, \forall i\notin k\cup l$. (2) For a client $j\in C_{k,l}$ we have
$x_{ij}=\frac{1-1/n^2}{n},   \forall  i\in   k$,  $x_{ij}=\frac{1}{n^3},
\forall i\in l$ and $x_{ij}=0, \forall i\notin k\cup l$.  (3) For a client
$j\notin  C_{k,l}$ we  have  $\forall i\in  k\cup  l$, $x_{ij}=0$  and
$\forall i\notin k\cup l$, $x_{ij}=\frac{1}{n^2-2n}$.
\end{definition}

We say that two vectors
 $s_{k,l},s_{k',l'} \in \mathcal{C}_I$ \emph{collide} with each other if 
$l \setminus (k' \cup l') \neq \emptyset$ and $l' \setminus(k\cup
l)\neq \emptyset$. We proceed by proving that for each  $s \in
\mathcal{C}_I$ the ratio of the number of the members of $\mathcal{C}_I$ that do not 
collide with $s$  to the number of the colliding members is exponentially small.

\begin{lemma}\label{lemma:estimate-conflicts}
For each $s_{k,l} \in \mathcal{C}_I$ let $\mathcal{U} \subseteq \mathcal{C}_I$ be the set of vectors in the core that collide with $s_{k,l}$. Then $\frac{|\mathcal{C}_I|-|\mathcal{U}|}{|{\mathcal{C}_I}|}=2^{-\Omega(n \log n)}$.
\end{lemma}

\begin{proof}
We lower-bound the ratio in question by upper bounding the probability
that a member  of $\mathcal{C}_I$ picked uniformly at  random does not
collide with  $s_{k,l}$. Consider  the event $\mathcal{E}_1$  that $l'
\setminus  (k\cup  l)   =  \emptyset$.  It  must  be   the  case  that
$l'\subseteq k\cup
l$.  The probability $P[\mathcal{E}_1]$  is at most
$(\frac{2n}{n^2})^n=(2/n)^n$ -- this is the probability
that all members of $l'$ are in $k\cup l$ if we were to pick them with
repetition and  the probability of the actual  $\mathcal{E}_1$ is less
since  we  do not  allow  repetitions in  the  set  $l$. Likewise  the
probability of the event $\mathcal{E}_2$ that $l \setminus (k'\cup l') =
\emptyset$ is the same. So, by the union bound, the probability that a 
randomly  picked  element of  $\mathcal{C}_I$  does  not collide  with
$s_{k,l}$ is $P[\mathcal{E}_1\cup \mathcal{E}_2]\leq 2(2/n)^n$.
\end{proof}

Next we show that for any two colliding vectors $s_{k,l}$ and
$s_{k',l'}$ in  $\mathcal{C}_I$
there is a convex combination $s'$ of them that is contained in  the integer polytope.

\begin{lemma}\label{lemma:conflict}
For any two colliding   vectors $s_{k,l},s_{k',l'}\in \mathcal{C}_I,$  
$\operatorname{conv}(\{s_{k,l},s_{k',l'}\}) \cap \operatorname{conv}(P) \neq \emptyset$.
\end{lemma} 

\begin{proof}
We will actually show that the average vector
$s'=\frac{s_{k,l}+s_{k',l'}}{2}$  is a  convex combination  of integer
solutions. We will do show by giving a distribution $\mathcal{D}$ over
integer    solutions    whose    expected   vector    $(y^\mathcal{D},
x^\mathcal{D})$ 
with respect to $\mathcal{D}$
is  $s'$. The  intuition behind  the  proof is  that each  one of  the
vectors $s_{k,l},s_{k',l'}$, in order to become a convex combination
of  integer solutions,  needs what  the other  has in  abundance: some
measure for the $y$ variable of a facility in $l$ and some measure for
the $y$ variable of a  facility in $l'$ respectively. With probability
$1/2$ we choose  to perform experiment $A$ and  with probability $1/2$
we perform experiment $B$ described below. 

Suppose that  experiment $A$  is chosen. We  will describe  the random
solution in two steps: $A_1$ and $A_2$.  We describe first step $A_1.$ 
Let $f$ be a member of the
set $l
\setminus (k'\cup l')$, which is non-empty by the choice of
$k,l,k',l'$. We select exactly one  facility to be opened from the set
$l$ according to to the following probabilities:  for $i\in l-\{f\}$
the probability is  equal to $y^{s_{k,l}}_i=\frac{10}{n^2}$, while for
facility   $f$    the   probability   is    equal   to   $1-\sum_{i\in
  l-\{f\}}y^{s_{k,l}}_{i}$. Facilities in $k$ are always opened in the
experiment step $A_1$. When some  facility $i\in l-\{f\}$ is chosen we
randomly               select              $w^i_{A_1}=\frac{\sum_{j\in
    C_{k,l}}x^{s_{k,l}}_{ij}}{y^{s_{k,l}}_i}$  clients  from $C_{k,l}$
and assign them to $i$ -- we assume without loss of generality that
$w^{i}_{A_1}$ is  an integer,  in the Appendix  we show how  to handle
fractional $w$'s.  Assign the remaining clients  in $C_{k,l}$ randomly
to  the facilities  $i'\in k$  so that  each one  is  assigned exactly
$w^{i'}_{A_1}=\frac{|C_{k,l}|-w^i_{A_1}}{n}$ clients  (again we assume
w.l.o.g. that  this is  an integer).  When  facility $f$ is  chosen we
randomly               select              $w^f_{A_1}=\frac{\sum_{j\in
    C_{k,l}}x^{s_{k,l}}_{fj}}{1-\sum_{i\in    l-\{f\}}y^{s_{k,l}}_{i}}$
clients from  $C_{k,l}$ and assign  them to $f$. Assign  the remaining
clients in $C_{k,l}$ randomly to the facilities $i'\in k$ so that each
one  is assigned  exactly $w^{i'}_{A_1}=\frac{|C_{k,l}|-w^f_{A_1}}{n}$
clients (again we assume w.l.o.g. that the $w$'s are integers).

For the second step $A_2$ of the experiment, let $g$ be a facility in
$l'\setminus (k\cup l)$.  We select facility $g$ to be opened with a
probability $\sum_{i\in  l}y^{s_{k,l}}_{i}$, the other  facilities in
$F-(k\cup l)$ are always opened in the experiment step $A_2$. If $g$
is            opened,             it            is            assigned
$w^g_{A_2}=\frac{\sum_{j}x^{s_{k,l}}_{gj}}{\sum_{i\in
    l}y^{s_{k,l}}_{i}}$ clients  randomly chosen from  $C-C_{k,l}$ and
the  remaining clients  of $C-C_{k,l}$  are assigned  randomly  to the
facilities $i'$  in $F-(k\cup l)-\{g\}$  so that each one  is assigned
exactly            $w^{i'}_{A_2}=\frac{|C-C_{k,l}|-w^g_{A_2}}{|F-(k\cup
  l)-\{g\}|}$. If  $g$ is not  opened, all the clients  in $C-C_{k,l}$
are assigned  randomly  to the
facilities $i'$  in $F-(k\cup l)-\{g\}$  so that each one  is assigned
exactly            $w^{i'}_{A_2}=\frac{|C-C_{k,l}|}{|F-(k\cup
  l)-\{g\}|}$.

Now suppose that  experiment $B$ is chosen. This  case is symmetric to
the  previous  experiment  by   exchanging  sets  $k,l$  with  $k',l'$
respectively  but  we  give  the  full description  for  the  sake  of
completeness.  Again we will describe the random solution in two steps
$B_1$ and $B_2$.  

We descibe first step $B_1.$ Let $g$ be the member of $l' \setminus (k\cup l )$
that we used in 
step $A_2$. We  select exactly one facility to be  opened from the set
$l'$ with respect to the following probabilities: for $i'\in l'-\{g\}$
the probability is equal to $y^{s_{k',l'}}_{i'}=\frac{10}{n^2}$, while
for   facility  $g$   the  probability   is  equal   to  $1-\sum_{i\in
  l'-\{g\}}y^{s_{k',l'}}_{i}$. Facilities in $k'$ are always opened in
the  experiment step  $B_1$. When  some facility  $i'\in  l'-\{g\}$ is
chosen     we    randomly     select    $w^{i'}_{B_1}=\frac{\sum_{j\in
    C_{k',l'}}x^{s_{k',l'}}_{i'j}}{y^{s_{k',l'}}_{i'}}$  clients  from
$C_{k',l'}$  and assign  them to  $i'$  -- assume  again w.l.o.g.  that
$w^{i'}_{B_1}$  is an  integer.  Assign  the rest  of  the clients  in
$C_{k',l'}$ randomly to the facilities $i''\in k'$ so that each one is
assigned   exactly  $w^{i''}_{B_1}=\frac{|C_{k',l'}|-w^{i'}_{B_1}}{n}$
clients  (again we  assume w.l.o.g.  that this  is an  integer).  When
facility      $g$      is      chosen     we      randomly      select
$w^{g}_{B_1}=\frac{\sum_{j\in
    C_{k',l'}}x^{s_{k',l'}}_{gj}}{1-\sum_{i\in
    l'-\{g\}}y^{s_{k',l'}}_{i}}$  clients from $C_{k',l'}$  and assign
them to $g$. Assign the rest of clients in $C_{k',l'}$ randomly to the
facilities  $i''\in   k'$  so  that  each  one   is  assigned  exactly
$w^{i''}_{B_1}=\frac{|C_{k',l'}|-w^{g}_{B_1}}{n}$  clients  (again  we
assume w.l.o.g. that the $w$'s are integers).

For the second step $B_2$ of the experiment, let $f$ be the facility in
$l \setminus (k'\cup l')$ used in step $A_1$. We select facility $f$ to
be  opened  with  a  probability $\sum_{i\in  l}y^{s_{k,l}}_{i}$,  the
other  facilities  in  $F-(k'\cup  l')$  are  always  opened  in  the
experiment step $B_2$.  If $f$ is opened, it is assigned
$w^f_{B_2}=\frac{\sum_{j}x^{s_{k',l'}}_{fj}}{\sum_{i\in
    l}y^{s_{k,l}}_{i}}$ clients randomly chosen from $C-C_{k',l'}$ and
the remaining  clients of $C-C_{k',l'}$  are assigned randomly  to the
facilities $i'$ in $F-(k'\cup l')-\{f\}$  so that each one is assigned
exactly         $w^{i'}_{B_2}=\frac{|C-C_{k',l'}|-w^f_{B_2}}{|C-(k'\cup
  l')-\{f\}|}$. If $f$ is not opened, all the clients in $C-C_{k',l'}$
are assigned  randomly  to the
facilities $i'$  in $F-(k'\cup l')-\{f\}$  so that each one  is assigned
exactly            $w^{i'}_{B_2}=\frac{|C-C_{k',l'}|}{|F-(k'\cup
  l')-\{f\}|}$.

It is easy to see that the outcome of each experiment is always a feasible
integer solution, since all  clients are assigned to opened facilities
and the capacities  are respected by the choice of  $w$'s.  It is also
easy to  verify that $s'$ is  the expected vector  of the distribution
$\mathcal{D}$ defined above. A facility $i \in F-(k\cup l \cup k' \cup
l')$    is   always    opened   in    both   experiments    and   thus
$y^{\mathcal{D}}_i=y'_i=1$. A facility $i' \in (k\cup l \cup k'\cup
l')   -(\{f,g\})$   is   opened   in   experiment   $A$   a   fraction
$y^{s_{k,l}}_{i'}$  of the  time and  is  opened in  experiment $B$  a
fraction  $y^{s_{k',l'}}_{i}$ of  time, and  since each  experiment is
selected       with       $1/2$       probability,       we       have
$y^{\mathcal{D}}_{i'}=\frac{y^{s_{k,l}}_{i'}+y^{s_{k',l'}}_{i'}}{2}=y^{s'}_{i'}$. For
facility $f$ we have that in experiment $A$ it is opened $1-\sum_{i\in
  l-\{f\}}y^{s_{k,l}}_i$  while   in  experiment  $B$   it  is  opened
$\sum_{i\in l}y^{s_{k,l}}_i=\sum_{i\in l'}y^{s_{k',l'}}_i$ of the time
and  so  $y^{\mathcal{D}}_f=1/2+y^{s_{k,l}}_f/2=y^{s'}_f$.  Similarly  for
facility $g$. By similar arguments the desired properties can be shown for the
assignment variables.  Let us consdier for example, for facility $f\in F-(k'\cup l')$
we have  that in  $A_1$ the  expected total demand  assigned to  it is
$w^{f}_{A_1}P[y_f=1]=\sum_{j\in  C_{k,l}}x^{s_{k,l}}_{fj}$ and  by the
symmetric way that  clients in $C_{k,l}$ are assigned  to $f$ in $A_1$
we   have  that  $x^{A_1}_{fj}=x^{s_{k,l}}_{fj}$   for  all   $j$.  In
experiment step  $B_2$, the expected  total demand assigned to  $f$ is
$w^{f}_{B_2}P[y_f=1]=\sum_{j\in C-C_{k',l'}}x^{s_{k',l'}}_{fj}$ and by
the symmetric way that clients in $C-C_{k',l'}$ are assigned to $f$ in
$B_2$ we have that $x^{B_2}_{fj}=x^{s_{k',l'}}_{fj}$ for all $j$. Thus
$x^{\mathcal{D}}_{fj}=\frac{x^{s_{k,l}}_{fj}+x^{s_{k',l'}}_{fj}}{2}$.
\end{proof}

\begin{theorem}\label{theorem::cfl}
Every approximate formulation for metric \cfl\ that uses the natural encoding and has integrality gap at most $g$ for some constant $g>0$, has $2^{\Omega (n\log n)}$ constraints.
\end{theorem}
\begin{proof}
We first prove that for every vector $s_{k,l} \in \mathcal{C}_I$
there   is   an  admissible   cost   function   $w_{k,l}$  such   that
$w_{k,l}^{T}s_{k,l}=o(w_{k,l}^Ts^{opt}_{k,l})$ where $s^{opt}_{k,l}$ is an optimal
integer    solution    of    $I(n^2,an^4,n^3)$   with    respect    to
$w_{k,l}$. Consider  two points $p_1,p_2$  in some Euclidean  space at
distance $1$ from each other. At the first point $p_1$, the facilities
of $k\cup l$ and the clients of $C_{k,l}$ are co-located and the remaining
facilities  and  clients are  all  co-located  at $p_2$.  Additionally
the facilities in $l$ have all opening cost of $1$ and the rest have $0$
opening cost.   It is easy  to see that  every integer solution  has a
cost of at  least $1$: either some client $j  \in C_{k,l}$ is assigned
to some facility located at $p_2$ and thus incurs a connection cost of
$1$, or some costly facility in $l$ must be opened at $p_1$, incurring a
facility cost of $1$. On the other hand $w_{k,l}^{T}s_{k,l}=o(1).$

Consider some  inequality $\pi$ of a $g$-approximate
relaxation $Q,$ where $g >0$ is a constant.  
(In fact the proof holds for $g=o(n)$).
Suppose there is some $s_{k,l}\in \mathcal{C}_I$  that
violates 
$\pi.$  Then, for every $s'_{k',l'}\in \mathcal{C}_I$ which
collides with $s_{k,l}$, $\pi$ must be satisfied otherwise by Lemma
\ref{lemma:conflict} we have violation of validity. By Lemma
\ref{lemma:estimate-conflicts} we have that $\pi$ eliminates 
$2^{-\Omega(n \log n)}|\mathcal{C}_I|$ members of the core, and by using the
union bound the theorem  is proved. We note that for the sake of simplicity  the
parameters are not optimized  -- by using a
different core we can get tighter bounds.
\end{proof}

\bibliographystyle{plain}
\bibliography{../SA_for_CFL/bibliography-ver1}

\appendix
\section{Appendix to Section~\ref{section:cfl}}

Here we explain how to handle  fractional bin capacities in the proof
of  Lemma~\ref{lemma:conflict}.

To handle the case where the  $w$'s are not integers (which is
actually always the case), we simply do the
following. We will give the proof  for step $A_1$ -- the proofs for the
other steps are similar. Each  time facility $f$ (or $i\in l-\{f\}$)
is  selected to be opened, the number of the clients that are randomly selected
 to be assigned to it is $\lfloor w^{f}_{A_1}  \rfloor$ ($\lfloor  w^i_{A_1} \rfloor$)  
with probability
$1-(w^f_{A_1}-\lfloor  w^f_{A_1}   \rfloor  )$  ($1-(w^i_{A_1}-\lfloor
w^i_{A_1} \rfloor  )$), otherwise the number of clients is $\lceil w^f_{A_1}
\rceil$($\lceil  w^i_{A_1} \rceil$).  If the  number of  clients assigned to $f$
($i$)  is  selected  to be  $\lfloor w^f_{A_1}  \rfloor$  ($\lfloor  w^i_{A_1}
\rfloor$)  then we  randomly select  $n(\frac{|C_{k,l}|-\lfloor  w^f_{A_1} \rfloor
}{n} - \lfloor (\frac{|C_{k,l}|-\lfloor  w^f_{A_1} \rfloor }{n}) \rfloor)$ (
$n(\frac{|C_{k,l}|-\lfloor    w^i_{i_{A_1}}   \rfloor    }{n}    -   \lfloor
(\frac{|C_{k,l}|-\lfloor  w^i_{A_1} \rfloor  }{n}) \rfloor)$)  facilities in
$k$ at random  and set the number of clients assigned to them $\lceil
\frac{|C_{k,l}|-\lfloor    w^f_{A_1}   \rfloor   }{n}    \rceil$   ($\lceil
(\frac{|C_{k,l}|-\lfloor w^i_{A_1}  \rfloor }{n} \rceil$) and set the number of
of clients assigned to the remaining facilities in $k$ to $\lfloor(\frac{|C_{k,l}|-\lfloor
  w^f_{A_1} \rfloor }{n} \rfloor$($\lfloor(\frac{|C_{k,l}|-\lfloor w^i_{A_1}
  \rfloor  }{n} \rfloor$).   Otherwise  select some  $n(\frac{|C_{k,l}|-\lceil
  w^f_{A_1} \rceil  }{n} - \lfloor  (\frac{|C_{k,l}|-\lceil w^f_{A_1} \rceil
}{n}) \rfloor)$ ( $n(\frac{|C_{k,l}|-\lceil  w^i_{A_1} \rceil }{n} - \lfloor
(\frac{|C_{k,l}|-\lceil w^i_{A_1} \rceil  }{n}) \rfloor)$) facilities in $k$
at  random and  set the number 
of clients assigned to them to  $\lceil
(\frac{|C_{k,l}|-\lceil    w^f_{A_1}    \rceil    }{n}   \rceil$    ($\lceil
(\frac{|C_{k,l}|-\lceil  w^i_{A_1} \rceil  }{n} \rceil$)  and set the number of
of clients assigned to the rest of them  $\lfloor(\frac{|C_{k,l}|-\lceil w^f_{A_1} \rceil }{n} \rfloor
$($\lfloor(\frac{|C_{k,l}|-\lceil w^i_{A_1} \rceil }{n} \rfloor$). Note that
 the  expected  vector  is as  in  the  proof of  Lemma
\ref{lemma:conflict}.

\end{document}